\documentclass[12pt]{article}

\usepackage{amsmath,amssymb,amsthm,color,epsfig,mathrsfs}
\usepackage{amsbsy}

\usepackage{tabularx,ragged2e,booktabs,caption}

\newtheorem{theorem}{Theorem}

\newtheorem{lemma}[theorem]{Lemma}

\newcounter{spslist}

\newcounter{geqncount}
    {\refstepcounter{equation}%
     \setcounter{geqncount}{\value{equation}}%
     \setcounter{equation}{0}%
  }%
    {\setcounter{equation}{\value{geqncount}}}

\newcommand{\CC}{\mathbb{C}}

\newcommand{\cchi}[1]{{\textstyle{\chi\big(#1\big)}}}

\begin{document}

\bibliographystyle{plain} 

\begin{center}
{\bfseries \Large  Spectra of half-infinite quantum graph tubes\footnote[1]{This paper was a senior thesis by Jeremy Tillay under the advisement of Professor Stephen Shipman. A journal article on this subject is in preparation.}}
\end{center}

\vspace{0ex}

\begin{center}
{\scshape \large Jeremy Tillay\\
\vspace{2ex}
{
\itshape
Department of Mathematics\\
Louisiana State University\\
Baton Rouge, Louisiana \ 70803, USA\\
Senior Honors Thesis
}}
\end{center}

\vspace{3ex}
\centerline{\parbox{0.9\textwidth}}
{\bf Abstract.}\
Carbon nanotubes are a feverishly-studied topic in the scientific community as of late. Mathematically, they can be modeled with a quantum graph. Here we consider a structure somewhat similar to carbon nanotubes, another quantum graph tube that is formed by rolling a square lattice instead of a graphene structure. This symmetry imposes properties that make certain motions easier to study by creating convenient pairs of incoming and outgoing motions.

\vspace{3ex}
\noindent
\begin{mbox}
\end{mbox}
\vspace{3ex}

\hrule
\vspace{1.1ex}

\section{Introduction} 

A quantum graph is a metric graph (a set of vertices and edges where each edge is parametrized by an interval ) equipped with a Schr\"{o}dinger operator that acts on functions defined along the edges of the graph. Consider a two-dimensional square lattice $Q$, treated as a quantum graph equipped with the second derivative operator 
$ H= \partial_{xx} \ $, which is a Schr\"{o}dinger operator with a zero potential. 

\centerline{\scalebox{0.3}{\includegraphics{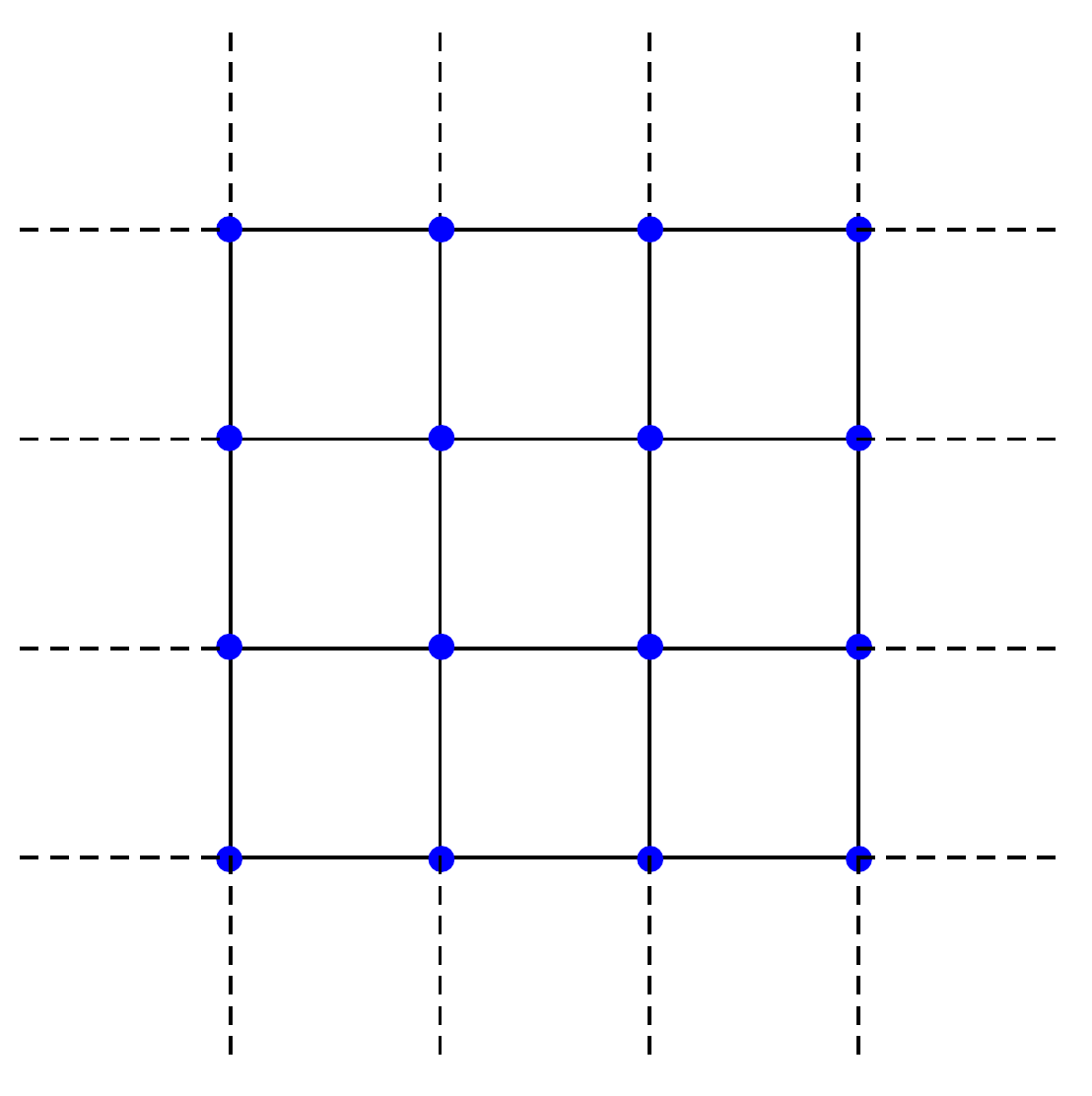}}}

Now we consider the property of self-adjointness. This operator H is self-adjoint with respect to the domain of functions if the equality $\langle Hu, v \rangle \ = \ \langle u, Hv \rangle$ holds for all functions within the specified domain. Note that the inner product of two functions $u$ and $v$ is considered here to be the $L^2[0, 1]$ norm acting on functions in $H^2[0, 1]$
\begin{equation}\label{condition1}
\langle u, v \rangle = \sum\nolimits_{e \in Q} \int_{0}^{1}uv\,dx
\end{equation}
So in order for the second-derivative operator H to be self-adjoint, it must be that
\begin{equation}\label{condition2}
\langle Hu, v \rangle = \langle u, Hv\rangle
\end{equation}
Or equivalently: 
\begin{equation}\label{condition3}
\sum\nolimits_{e \in Q} \int_{0}^{1}u''vdx= \sum\nolimits_{e \in Q} \int_{0}^{1}v''udx
\end{equation}
Through intergration by parts, expression (\ref{condition3}) can be simplified to:
\begin{equation}\label{condition4}
\sum\nolimits_{e \in Q}u'v(1)-u'v(0)= \sum\nolimits_{e \in Q} v'u(1) - v'u(0) 
\end{equation}

Now consider all functions satisfying continuity and vertex conditions such that at each vertex $x$, $u_e(x) = u_{e'}(x)$ where $e$ and $e'$ are any edges that meet at $x$
and the outgoing derivatives from the vertex satisfy: 
$\sum\nolimits_{e} u'_e (x) = 0$

Keep in mind that the derivatives must be normalized with a positive or negative sign to assure they describe outgoing motion. For example $u_e(1) $is the ingoing energy at some vertex parametrized by 1 along that edge $e$. So $-u_e(1)$ actually defines the energy outgoing from that point.

It is now true that at each vertex $\sum\nolimits_{e} u'v(x)  = v(x)\sum\nolimits_{e} u'_e (x) = v(x)\cdot0 = 0 $ . So both sides of (\ref{condition4}) are 0 at each vertex. Summing over all vertices in the graph makes it clear that the self-adjointness condition is satisfied for all functions satisfying these continuity and flux conditions. 

Since we are modeling free vibrations on these tubes, we can assume that there will be a wave motion along each edge, so it will be a linear combination of $ \sin{kx} $ and 
$ \cos{kx} $. This $k$ is constant through the graph, as it corresponds to innate qualities like tension, material composition of the wires that make up the square sheet, air quality surrounding the tube. Since this tube is considered to be perfectly symmetric and lie in a vacuum, the value $k$ does not change between edges.

Note this also makes $u$ an eigenfunction with eigenvalue $k^2$. If $u_e(x) = A \sin{kx}+ B \cos{kx}$ along each edge, then $Hu = -d_{xx} u = k^2 u(x)$

\centerline{\scalebox{0.3}{\includegraphics{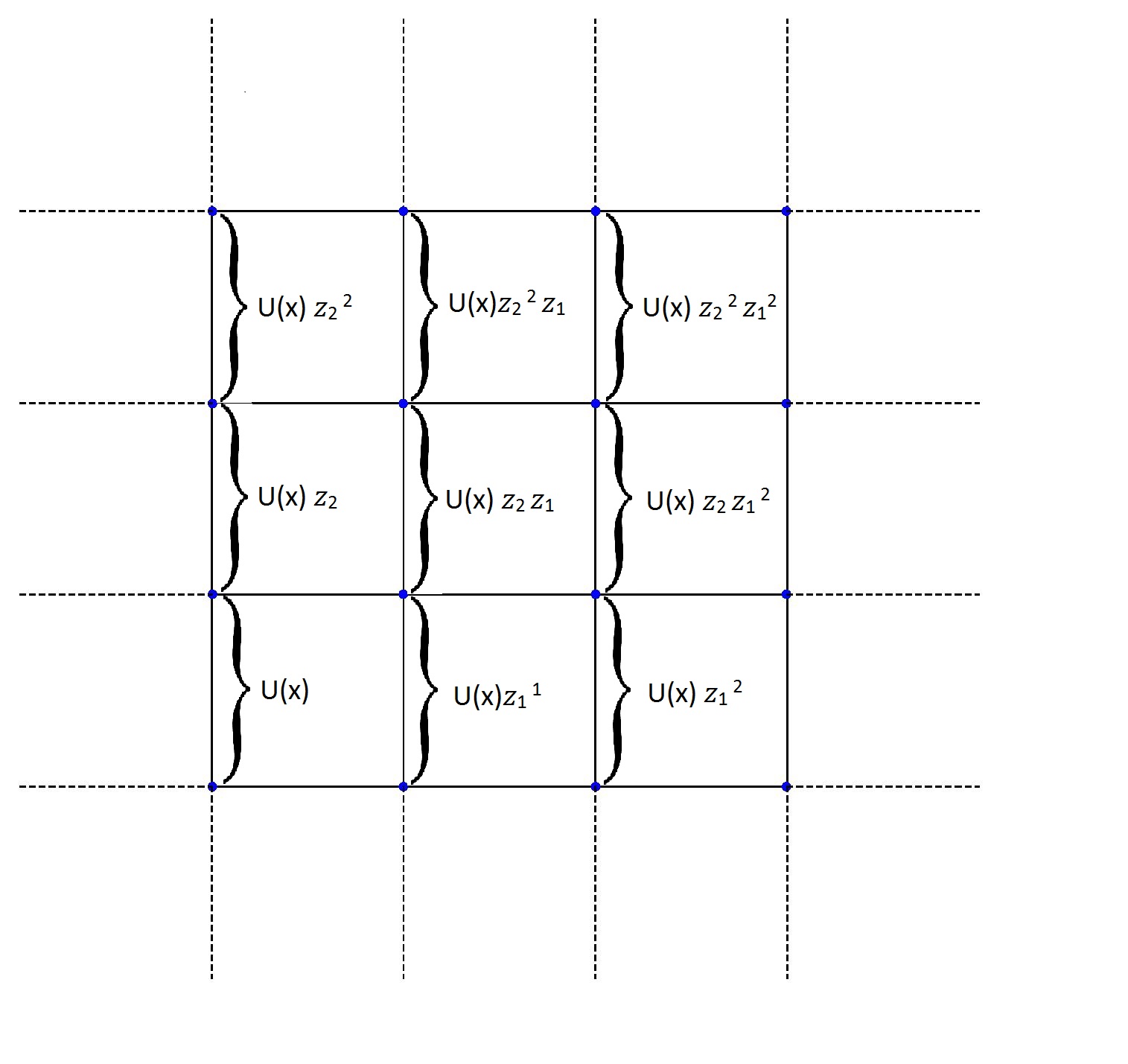}}}

Finally we consider the notion of Floquet multipliers, which will complete our construction of this quantum sheet. Since this structure is perfectly symmetric, it would make sense if there was a consistent phase and amplitude shift when moving vertically and horizontally. So we consider just that, a pair of multipliers  $z_2$ corresponding to vertical motion upwards and $z_1$ corresponding to horizontal motion to the right. Thus with knowledge of just the domain pictured below (called the fundamental domain), consisting of one vertex and two perpendicular edges, we can multiply by powers of $z_1$ and $z_2$ to get the corresponding vibrations at at any point in the graph.

\centerline{\scalebox{0.3}{\includegraphics{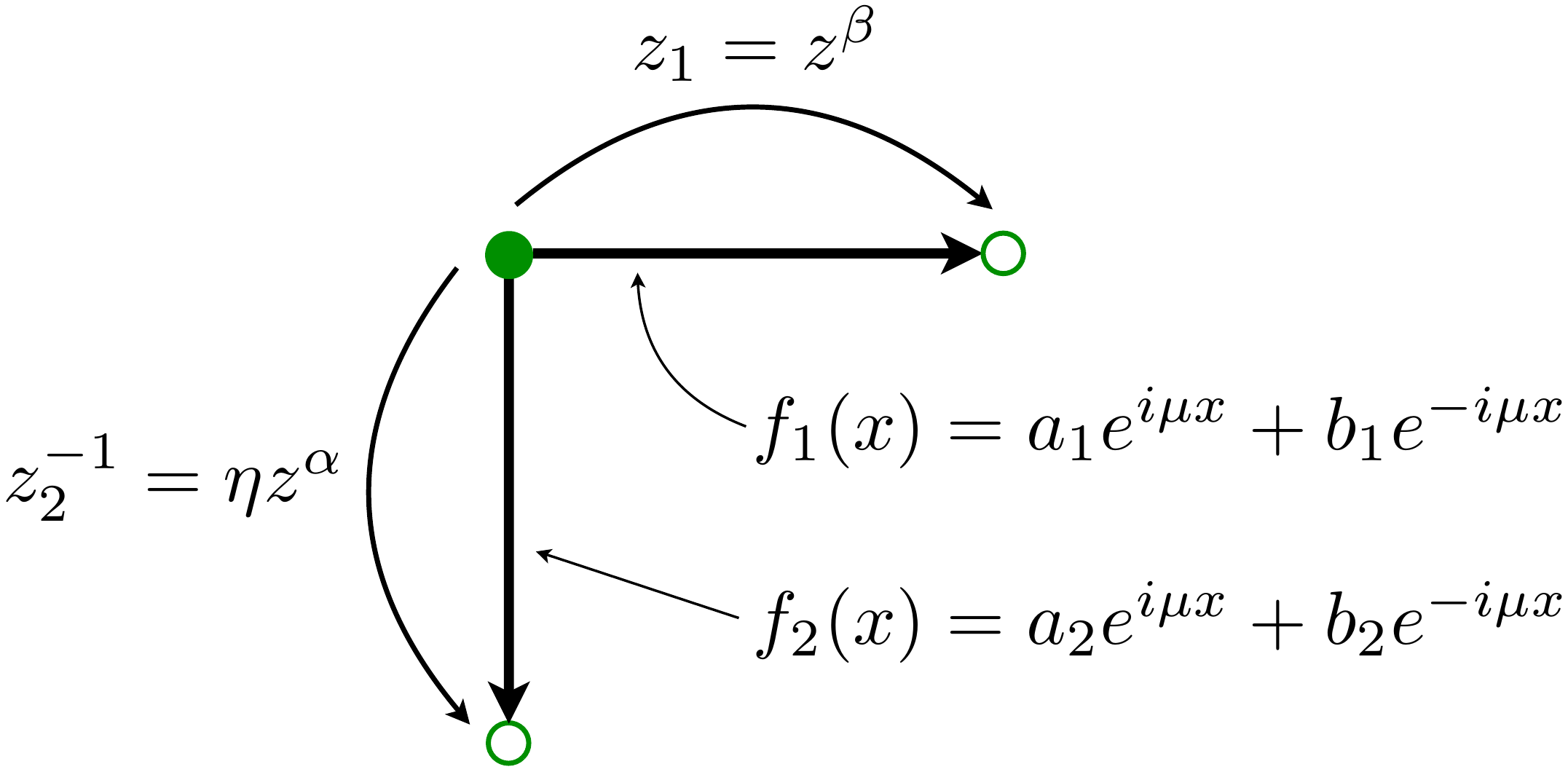}}}

Note that for ease of  calculations, we consider for most of this paper that the motion on the edges is a linear combination of $e^{ikx}$ and $e^{-ikx}$ instead of sine and cosine. The two representations of motion are equal up to an isomorphism. 

Combining all of the things just discussed, including our vertex conditions, Floquet multipliers, and assumption that each edge is a combination of sine and cosine motions, we can definitively solve for what combinations of motions and what multipliers are possible.

Now assume that the horizontal edge has vibrational equation $a_1e^{ikx} + b_1e^{-ikx}$ where $x = 0$ at the solid green node and approaches 1 as it moves right. The vertical edge has vibrational equation $a_2e^{ikx} + b_2e^{-ikx}$ where, again, $x = 0$ at the solid green node and approaches $1$ as it moves down.
 
 This leads to the following 4 equations: 3 matching and 1 flux condition as follows

\begin{equation} \label{lineqn1}
a_1 + b_1 = a_2 + b_2 
\end{equation}
\begin{equation} \label{lineqn2}
a_1e^{ik} + b_1e^{-ik}= z_1 (a_1 + b_1)
\end{equation}
\begin{equation}\label{lineqn3}
a_2 e^{ik} + b_2e^{-ik}= z_2 (a_2 + b_2)
\end{equation}
\begin{equation} \label{lineqn4}
a_1 - b_1 + a_2 - b_2 -  z_2 (a_2 e^{ik} - b_2e^{-ik}) - z_1^{-1}(a_1e^{ik} - b_1e^{-ik}) = 0.
\end{equation}

Conditions (\ref{lineqn1}-\ref{lineqn4}) are matching conditions, where continuity must be established at the solid green node. Condition 4 is the flux condition. 

It is clear that if these conditions are met on this fundamental domain, they will be met everywhere. Every other vertex will have values that are just scaled by a constant. So if there is continuity at this vertex, there will be continuity still after we multiply each of equations by a constant. The same goes for the flux condition. Multiplying both sides of the equations by $z_1^k z_2^l$ will not suddenly make their equality untrue. So if matching and flux are met on this vertex, they are met by all vertices. 

These conditions can be expressed as a linear equation of the form

\begin{equation}
\renewcommand{\arraystretch}{1.3}
\left[
  \begin{array}{cccc}
    1 & 1 & -1 & -1 \\
    z_1-e^{ik} & z_1\!-\!e^{-ik}  & 0 & 0 \\
    0 & 0 & z_2^{-1}\!-\!e^{ik}  & z_1^{-1}\!-\!e^{-ik}  \\
    1\!-\!e^{ik}  z_1^{-1} & 0 & 1\!-\!e^{ik} z_2 & 0
  \end{array}
\right]
\left[
  \begin{array}{c}
    a_1 \\ b_1 \\ a_2 \\ b_2  
  \end{array}
\right]
=
\left[
  \begin{array}{c}
    0 \\ 0 \\ 0 \\ 0
  \end{array}
\right]
\end{equation}

This is a 4x4 matrix with determinant:

$ \sin{k} * (z_1 + z_1^{-1} + z_2 + z_2^{-1}-4 \cos{k})$ 

Note that if $ \sin{k} = 0$, then $k = l*\pi$ for some integer l. In other words, $e^{ik} = (-1)^l$. So the movement from one vertex to one on the opposite side of the edge is just 1 or -1. This is a special case, and we will assume that $k$ is not one of these values for the duration of the paper, but such a possible setup is worth noting regardless.

We will soon seek solutions ${z_1, z_2}$ such that $z_1 + z_1^{-1} + z_2 + z_2^{-1}= 4 \cos{k}$, where, when satisfied, the coefficients are 

\begin{equation}\label{coefficients}
  \renewcommand{\arraystretch}{1.1}
\left[
  \begin{array}{c}
      a_1 \\ b_1 \\ a_2 \\ b_2
  \end{array}
\right]
=\,
c \left[
  \begin{array}{c}
      z_1-\zeta^{-1} \\ -z_1 + \zeta \\ z_2^{-1} - \zeta^{-1} \\ -z_2^{-1} + \zeta
  \end{array}
\right].
\end{equation}

First, however, we must consider the wrapping of this sheet in order to model a tube.

One can easily picture how to wrap this tube. Picture being an ant sitting on this sheet. Now pick a direction on the 2-D square sheet represented by the vector $\langle \alpha, \beta \rangle$. Face that direction and imagine someone cuts off all of the sheet behind directly behind you. Now start marching and at some arbitrary point in your path, I will roll the sheet so that the point where you started is glued directly under the point on your path that I chose. Then I will roll the entire sheet in the same pattern, so that you will still think you are marching in a straight line, but I have really made it so that you are moving along the tube and are periodically returning to the same point you started, though on a different layer of the tube. 

Note: This folding will have to perfectly align two points on your path or you will spiral down the tube forever. 

Now imagine instead of traveling on a square plane, you are on this quantum sheet with wires. You don't want to fall off, so you travel in a fixed pattern along vertical and horizontal edges. So you move from one edge to the edge directly above you $\alpha$ times, and then to the edge directly to the right $\beta$ times. Downward or leftward motion corresponds to $-\alpha$ or $-\beta$ respectively. 

Now wrap the tube so that the points $(0, 0)$ and $(\alpha, \beta)$ are "glued" together(and furthermore, all points of the form $(x + m\alpha, y + m\beta), m \in \mathbb{Z}, (x, y) \in Q$ are "glued" together). On this infinitely thin tube, since there is continuity, each of the points on the tube must be vibrating in exactly the same way now that we have "glued" them. Since we have an assumption about Floquet multipliers, the mathematical expression of the physical idea of gluing or folding the tube is naturally defined as follows: 

$u(x + \alpha, y + \beta) = u(x,y)$ is the new matching or "gluing" relation and
$u(\alpha, \beta) = u(0, 0)z_1^{\alpha} z_2^{\beta}$ is the Floquet multiplier condition.

This must hold not just at the vertices, but at each point along every edge. So even if $u(x, y) = u(x + \alpha, y+ \beta) = 0$, at some point (x, y) in the grid, then this equality will not hold along the entire graph unless the function $u$ is zero throughout the entire graph (a trivial case) or $z_1^{\alpha}z_2^{\beta} = 1$

Now we have two unknowns, $z_1, z_2$ and two conditions:
\begin{equation}\label{zequation1}
\begin{array}{c}
z_1^\alpha z_2^\beta = 1 \\
z_1 + z_1^{-1}+ z_2 + z_2^{-1} = 4 \cos{k}.
\end{array}
\end{equation}

This is not a linear system, but it turns out we can still solve it for a set of $2\beta$ pairs of Floquet multipliers  ${z_1, z_2}$ 

It is convenient to assume that $gcd(\alpha, \beta) = 1$ and write the above equality as $z_1^{d\alpha}z_2^{d\beta}$.

The condition $z_1^{\alpha\delta} z_2^{\beta\delta} = 1$ is equivalent to 
\begin{equation}
  z_1^\alpha z_2^\beta = \cchi{\frac{\ell}{\delta}}
  \quad \text{for some integer } \ell:\,0\leq\ell<\delta,
\end{equation}
in which $\chi(t) = e^{2\pi it}$.  Thus (\ref{zequation1}) is equivalent to the existence of $\ell:\,0\leq\ell<\delta$ such that
\begin{equation}\label{zequation2}
  \renewcommand{\arraystretch}{1.1}
\left\{
  \begin{array}{l}
    z_1 + z_1^{-1} + z_2 + z_2^{-1} = 4\cos\mu \\
    z_1^\alpha z_2^\beta = \chi\!\left( \frac{\ell}{\delta} \right). \\
  \end{array}
\right.
\end{equation}

The following lemma characterizes the solutions of (\ref{zequation2}).

\begin{lemma}\label{lemma:zequation}
If $\alpha$ and $\beta$ are positive integers with $\gcd(\alpha,\beta)=1$ and $(z_1,z_2)\in(\CC^*)^2$, then (\ref{zequation2}) holds if and only if there exists $z\in\CC^*$ such that
\begin{equation}\label{zequation3}
  \renewcommand{\arraystretch}{1.1}
\left\{
  \begin{array}{l}
    z^\beta + z^{-\beta} + \eta z^\alpha + \eta^{-1} z^{-\alpha} = 4\cos\mu\,,\; \text{ with } \eta=\cchi{\frac{-\ell}{\beta\delta}} \\
    z_1 = z^{\beta} \\
    z_2 = \eta^{-1} z^{-\alpha}\,.
  \end{array}
\right.
\end{equation}
Such $z$ is unique.  The same pair $(z_1,z_2)$ satisfies the modification of the system (\ref{zequation3}) by the replacements
\begin{eqnarray}
  z \mapsto \cchi{\frac{j}{\beta}}z,
  \quad
  \eta \mapsto \cchi{\frac{-j\alpha}{\beta}}\eta,
\end{eqnarray}
featuring the isomorphism $\cchi{\frac{j}{\beta}}\mapsto\cchi{\frac{-j\alpha}{\beta}}$ of the group of $\beta^\text{th}$ roots of $1$.

If $z$ satisfies the first equation in the system (\ref{zequation3}) and $|z|\not=1$, then $z$ is a simple root of the equation.
\end{lemma}

\begin{proof}
To prove that (\ref{zequation3}) implies (\ref{zequation2}) is straightforward.
  To prove the uniqueness of the number $z\in\CC^*$ that satisfies (\ref{zequation3}), suppose $\eta^{-1}z^{-\alpha}=\eta^{-1}w^{-\alpha}$ and $z^\beta=w^\beta$ for some $z$ and $w$ in $\CC^*$.  Then $(zw^{-1})^\alpha=1=(zw^{-1})^\beta$.  Since $\gcd(\alpha,\beta)=1$, $zw^{-1}=1$, so that $z=w$.
  
  To prove the existence of such $z$ under the assumption of (\ref{zequation2}), suppose that $z_1^\alpha z_2^\beta=\zeta:=\cchi{\frac{\ell}{\delta}}$, and let $\eta$ be such that $\eta^{-\beta}=\zeta$.  Then choosing numbers $w_1$ and $w_2$ such that $z_1=w_1^\beta$ and $z_2=\eta^{-1}w_2^{-\alpha}$ yields $(w_1w_2^{-1})^{\alpha\beta}=1$ and therefore $w_1w_2^{-1}=\cchi{\frac{r}{\alpha\beta}}$ for some integer $r$.  Since $\gcd(\alpha,\beta)=1$, there exist integers $m$ and $n$ such that $r=-m\alpha + n\beta$, or
\begin{equation}
  \frac{r}{\alpha\beta} = -\frac{m}{\beta} + \frac{n}{\alpha}\,.
\end{equation}
This implies that
\begin{equation}
  w_1w_2^{-1} = \cchi{\frac{-m}{\beta}} \cchi{\frac{n}{\alpha}},
\end{equation}
so that $w_1\cchi{\frac{m}{\beta}}=w_2\cchi{\frac{n}{\alpha}}:=z$.  This is the desired number since
$z^\beta=w_1^\beta=z_1$ and $\eta^{-1}z^{-\alpha}=\eta^{-1}w_2^{-\alpha} = z_2$ and first equation of (\ref{zequation2}) becomes the first equation of (\ref{zequation3}).

The condition for a root $z$ of the Laurent polynomial in (\ref{zequation3}) to be a multiple root is
\begin{equation}\label{multipleroot}
  \beta\left( z^\beta - z^{-\beta} \right) = -\alpha\left( \eta z^\alpha - \eta^{-1} z^{-\alpha} \right).
\end{equation}
Both sides of this equation lie on ellipses in $\CC$ with vertical major axis.
The left side lies on an ellipse with major radius $\beta\left( |z|^\beta + |z|^{-\beta} \right)$ and minor radius $\beta\left| |z|^\beta - |z|^{-\beta} \right|$, and the right side lies on an ellipse with major radius $\alpha\left( |z|^\alpha + |z|^{-\alpha} \right)$ and minor radius $\alpha\big| |z|^\alpha - |z|^{-\alpha} \big|$.  The assumptions $\beta>\alpha$ and $|z|\not=1$ imply
$\beta\left( |z|^\beta + |z|^{-\beta} \right)>\alpha\left( |z|^\alpha + |z|^{-\alpha} \right)$
and $\beta\big| |z|^\beta - |z|^{-\beta} \big|>\alpha\big| |z|^\alpha - |z|^{-\alpha} \big|$,
so that the two ellipses do not intersect.  Thus (\ref{multipleroot}) cannot hold and $z$ is therefore a simple~root.
\end{proof}

\begin{theorem}\label{thm:zequation}
If $\alpha$ and $\beta$ are positive integers with $\gcd(\alpha,\beta)=1$ and $\delta$ is a positive integer, then
the set of solutions $(z_1,z_2)\in(\CC^*)^2$ to the system (\ref{zequation1}) is the {\em disjoint} union
\begin{equation}
  \bigcup\limits_{\ell=0}^{\delta-1} {\mathcal Z}_\ell
\end{equation}
of the solutions sets of (\ref{zequation2}),
\begin{equation}\label{union}
  {\mathcal Z}_\ell \,=\,
  \left\{ (z^\beta, \eta^{-1}z^{-\alpha})\,:\,
    z^\beta + z^{-\beta} + \eta z^\alpha + \eta^{-1} z^{-\alpha} = 4\cos\mu,\; z\in\CC^*,\; \eta=\cchi{\frac{-\ell}{\beta\delta}}
  \right\}.
\end{equation}
\end{theorem}

\begin{proof}
Because the system (\ref{zequation1}) is equivalent to the existence of an integer $\ell:0\leq\ell<\delta$ such that (\ref{zequation2}) holds, Lemma~\ref{lemma:zequation} establishes the union (\ref{union}).  To prove that it is disjoint, suppose that
\begin{equation}
  \renewcommand{\arraystretch}{1.1}
\left.
  \begin{array}{r}
    z_1=z^{\beta} = w^{\beta} \\
    z_2 = \eta^{-1} z^{-\alpha} = \nu^{-1} w^{-\alpha}
  \end{array}
\right\}
\quad\text{with $\eta=\cchi{\frac{-\ell}{\beta\delta}}$ and $\nu=\cchi{\frac{-k}{\beta\delta}}$}
\end{equation}
and $0\leq\ell<\delta$ and $0\leq k<\delta$.
One has $\eta\nu^{-1}=(wz^{-1})^\alpha$ and $(wz^{-1})^\beta=1$, so that
$(\eta\nu^{-1})^\beta = (wz^{-1})^{\alpha\beta} = 1^\alpha = 1$, and therefore $\cchi{\frac{k-\ell}{\delta}}=1$.  Since $|k-\ell|<\delta$, one has $\left| \frac{k-\ell}{\delta} \right|<1$, which together with $\cchi{\frac{k-\ell}{\delta}}=1$ yields $\frac{k-\ell}{\delta}=0$ and hence $\eta=\nu$.
\end{proof}

For the duration of our discussion, we assume $gcd(\alpha, \beta) = 1$ and, without loss of generality, $\beta \geq \alpha$

Now that we have folded the tube and assured that is satisfies certain mathematial expressions of physical reality, we are interested in imposing a defect to allow scattering and some potentially interesting interactions. The natural defect to impose is simply a cut along the vector of periodicity $\langle \alpha, \beta \rangle$ 

\centerline{\scalebox{0.3}{\includegraphics{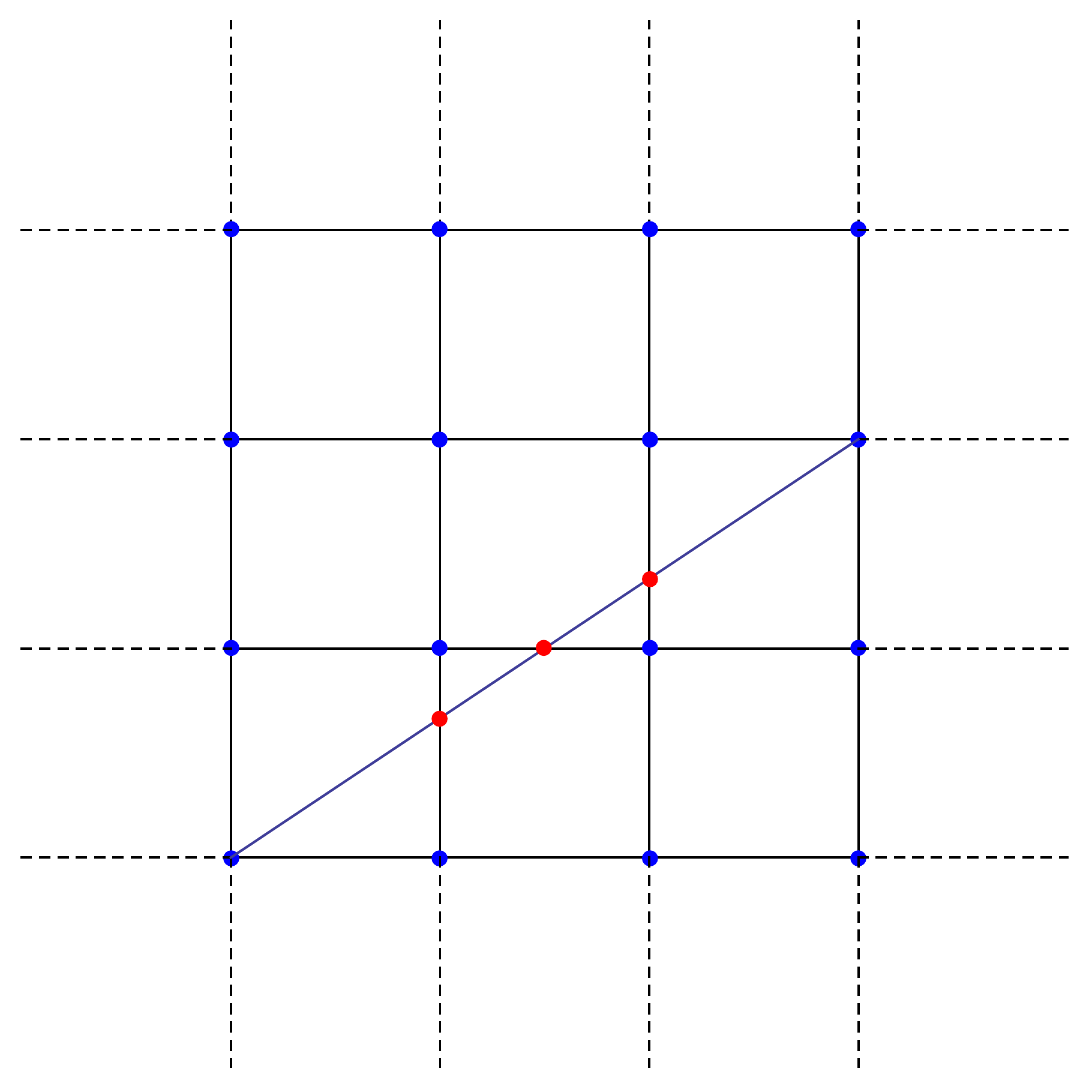}}}

In doing so, the flux condition is no longer satisfied at the three red dots and two blue dots on the line because the graph ceases to exist above and to the left of the cut. So we must impose new conditions that model a tube existing in a closed system, so that no energy is leaking out or spontaneously being created. Such a physical idea can be expressed with the following condition: 

Let $U$ be a unitary matrix. $F$ is a vector of values at the nodes as pictured above and $F'$ is a vector of the derivative at those nodes. 

$U\mathbf{(F + iF')} = \mathbf{F-iF'}$ 

If this is satisfied for some unitary matrix U, then the boundary conditions at the new vertices where we cut the tube are said to be self-adjoint.

Note that this holds true for some unitary matrix if and only if 

\[
\lVert \mathbf{F+iF} \rVert =\lVert \mathbf{F-iF} \rVert 
\]

This is a set of $\alpha + \beta$ boundary conditions because there will be $\alpha -1$ vertical intersections between (0, 0) to $(\alpha, \beta)$ and $\beta-1$ vertical intersections as well as 2 nodes to consider at (0, 0) along the horizontal edge and $(\alpha, \beta)$ at the vertical edge .

However, there will only be $2\beta$ modes coming from the solution to the dispersion and tube relations earlier. We take $\beta$ modes of the form ${z_1, z_2}$ to be input motions and the corresponding $\beta$ modes of the form ${z_1^{-1}, z_2^{-1}}$ motions to be output. This leaves a linear equation with $\beta$ unknowns and $\alpha+\beta$ conditions. This is not necessarily solvable, so we must find $\alpha$ degrees of freedom. 

\centerline{\scalebox{0.3}{\includegraphics{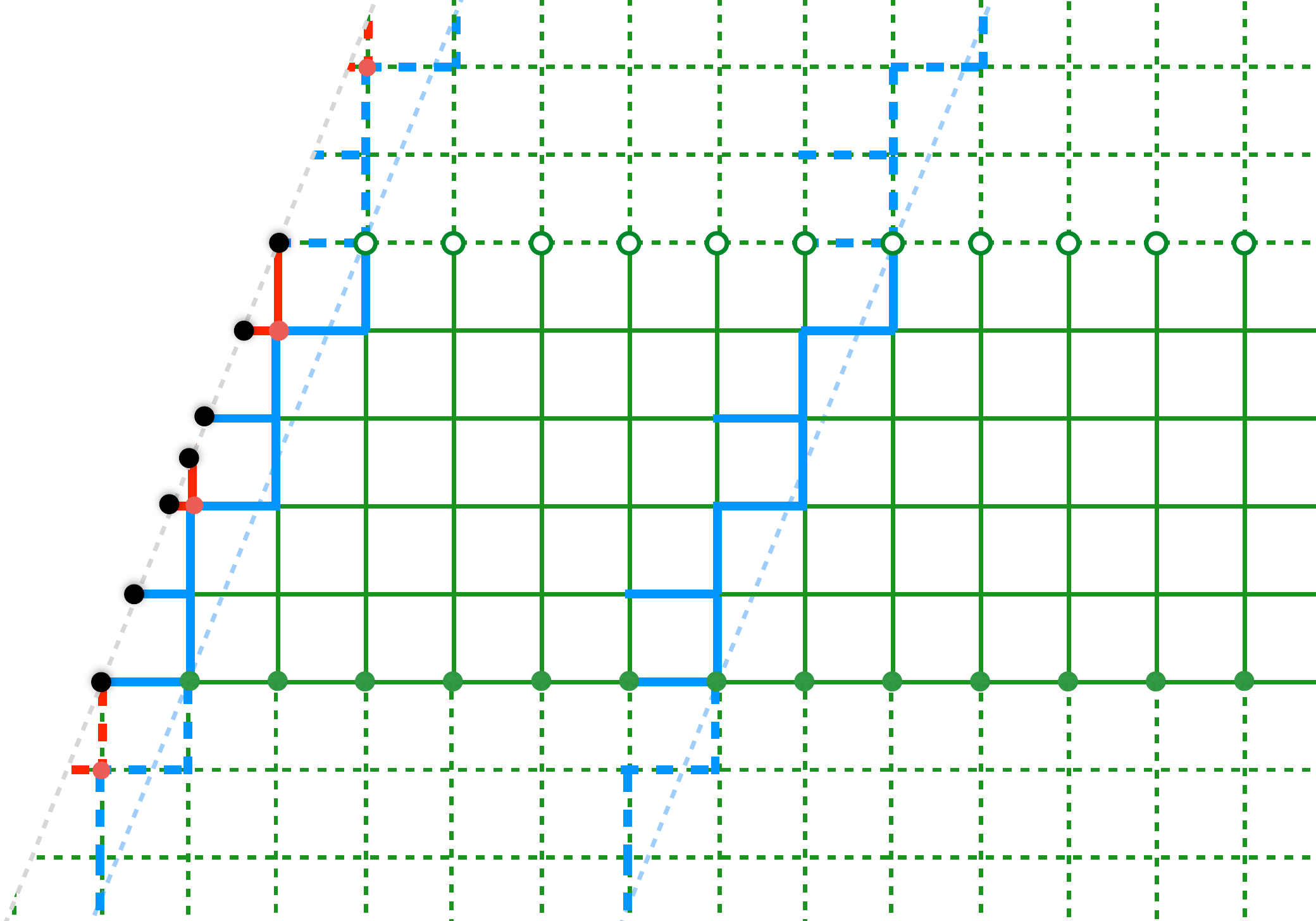}}}
\centerline{\scalebox{0.3}{\includegraphics{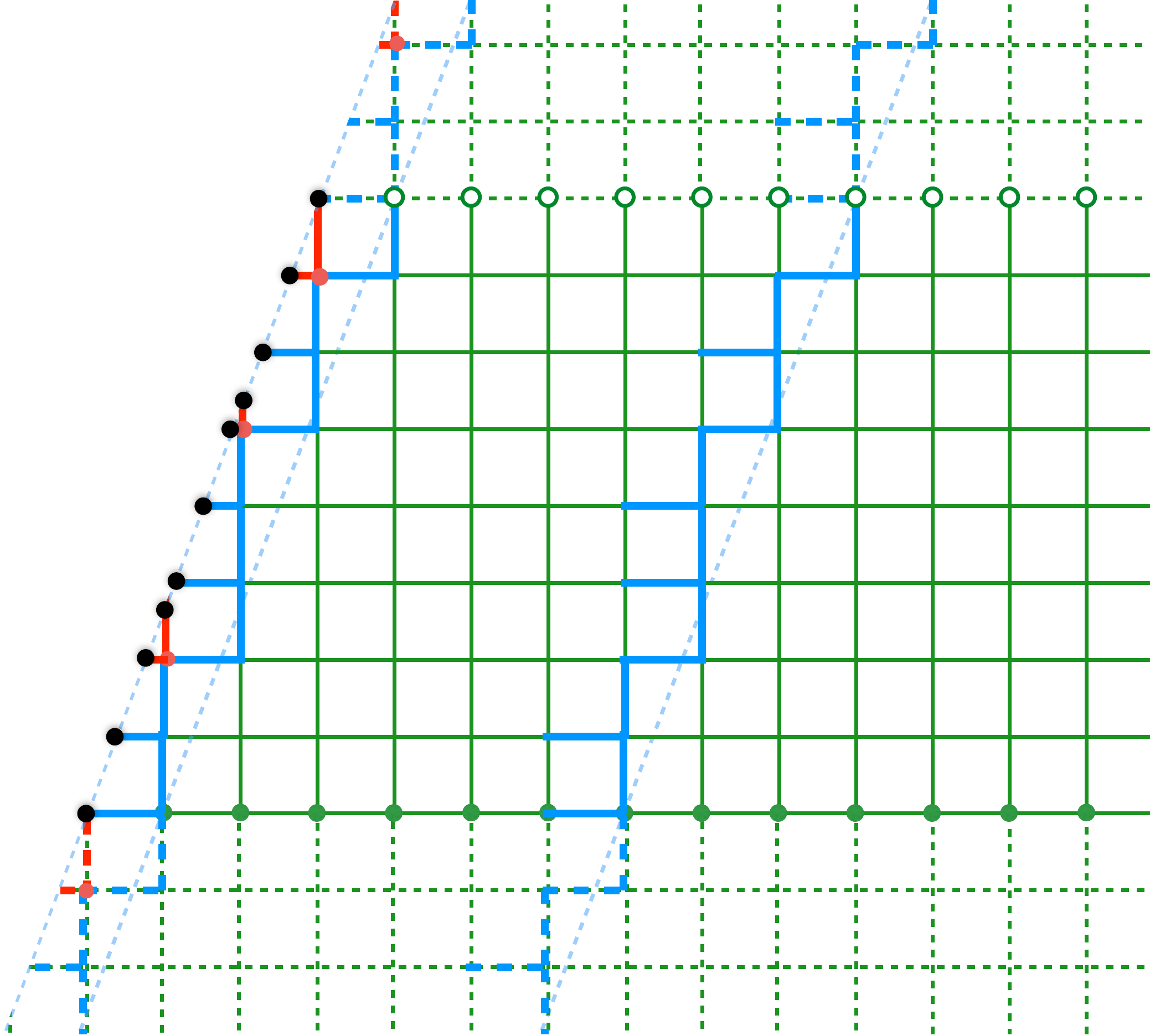}}}
\centerline{\scalebox{0.3}{\includegraphics{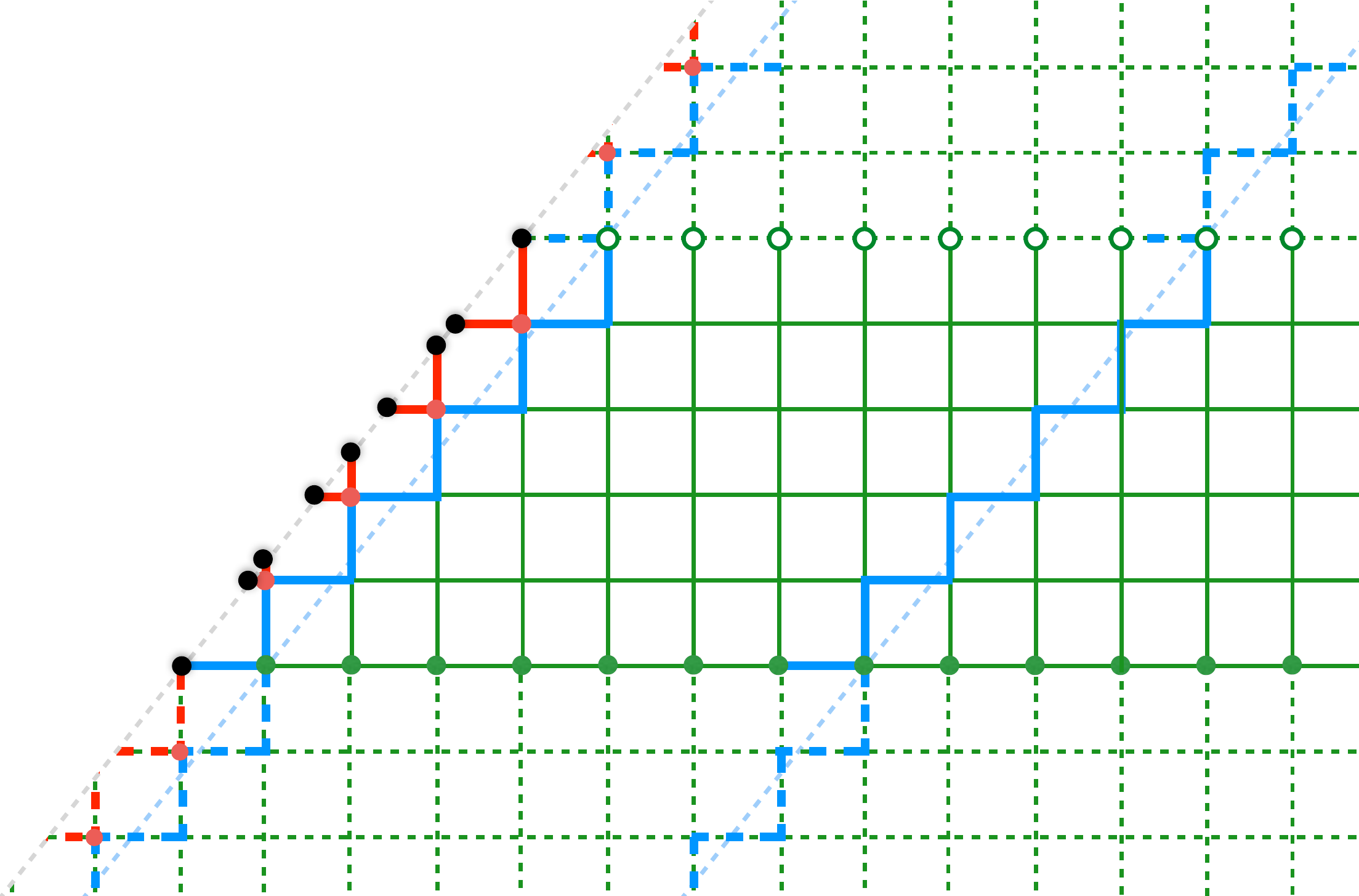}}}

We do so by allowing $\alpha$ auxillary modes as pictured along the red edges above

Note that the vertex conditions are still satisfied at the inside vertex of these edges with auxillary motion. The auxillary motion is $c_i \sin{kx}$ and $\-c_i \sin{kx}$ on the horizontal and vertical edges respectively, parametrized so that $x = 0$ at the interior vertex and approaches 1 as one moves along the red edge towards the edge at or past the boundary. At $x = 0$, this will just be 0 and 0, meaning matching will still be satisfied. The derivatives will be negatives of each other, so the total flux will still be 0.

Now we can create a linear system of equations that allow us to solve for these unknowns with some set of conditions encoded by the unitary matrix $U$. Let $X_i$ be the vector that is of the form 

$[0, 0, ... , \sin{kx_1}, -\sin{kx _2} , 0 , ... 0]^T$
   
So that it encodes the sine motion only on the corresponding auxillary edge at the boundary points encoded by an appropriate $x_1, x_2 \in [0, 1]$. Also let $F_i$ refer to the motion encoded by some ${z_1, z_2}_i$ Floquet multiplier where by convention the first $\beta$ solutions are rightward or response motions and the last $\beta$ solutions are the inverse leftward or source coefficients.

$([U-I] * [F_1, F_2, ...,F_\beta, X_1, 	... X_\alpha]  + i[U+I]* [F_1', F_2', ...,F_\beta', X_1, 	... X_\alpha])* 
\begin{bmatrix}
    R_1 \\
    R_2 \\
    \vdots     \\
    R_\beta \\
    C_1 \\
    C_2 \\
    \vdots     \\
    C_\alpha
\end{bmatrix}
$


 =  $\sum_{n=1}^{\beta} J_n([-U+I] * [F_{n+\beta}] + i(-U-I)*[F_{n+\beta}']$

\section{Conclusion and Further Work}

We must additionally consider when the matrix equation above has a solution when all the source coefficients $J_n$ are zero, or when the matrix multiplying our unknowns has a determinant of $0$. This can most likely be seen with experimental computations when the unitary matrix has some values that can be perturbed without violating the property that $U$ is unitary. If such a setup is possible, we have found a system of trapped energy, which implies that small perturbations may create a system with resonant interactions.

However, we can solve this system in general for an arbitrary input of energy to find how the system will react to predetermined boundary conditions encoded by $U$. Thus we can create a physically sensible mathematical model for a half-cut tube.

If the above is possible, perhaps a similar method of study could be applied to carbon nanotubes. This may lead the way to allowing resonant interactions in these often-studied structures. So if a resonant setup can be purposefully created, the naturally sensitive nature of a resonant system may make it possible to create very fast on-off switches in computers where the carbon nanotubes (which are semiconductors) can replace silicon. Of course, this is a far-off goal, but worth considering. 

\bigskip

\bibliography{QGTubes}

\end{document}